\documentclass[11pt,oneside]{article}

\usepackage{palatino}
\usepackage[margin=1in]{geometry}
\usepackage{color,xcolor,soul}
\usepackage[abspage,user,savepos]{zref}
\usepackage[colorlinks, citecolor=blue]{hyperref}  
\usepackage{natbib}
\usepackage{fp}
\usepackage{framed}
\usepackage{cuted,balance}
\usepackage{setspace}
\usepackage{graphicx} % for pdf, bitmapped graphics files
\usepackage{epsfig} % for postscript graphics files
\usepackage{mathtools}
\usepackage{dsfont,amsmath} % assumes amsmath package installed
\usepackage{amssymb}  % assumes amsmath package installed
\usepackage{algorithm}
\usepackage{algpseudocode}
\usepackage[capitalize]{cleveref}

\newcommand{\belief}{\pi} % We can make it \hat{\pi} as well but it may be better not to complicate the notation.
\newcommand{\argmax}{\mathop{\mathrm{argmax}}}
\newcommand{\nn}{\nonumber}
\newcommand{\Q}{Q}
\newcommand{\stepA}{\bar{\alpha}}

\newcommand{\vfunc}{\upsilon}
\newcommand{\ufunc}{u}
\newcommand{\deltaQ}{\delta\Q}

\newcommand{\XiV}{\Upsilon}
\newcommand{\XiU}{Y}
\newcommand{\uu}{\underline{\ufunc}}
\newcommand{\ue}{\underline{e}}
\newcommand{\tbeta}{\tilde{\beta}}

\newcounter{theorem}
\newenvironment{definition}{\refstepcounter{theorem}\par\medskip
   \noindent 
   \textbf{Definition \thetheorem.} \em \rmfamily}
\newenvironment{proposition}{\refstepcounter{theorem}\par\medskip
   \noindent  
   \textbf{Proposition \thetheorem.} \em \rmfamily}
\newenvironment{lemma}{\refstepcounter{theorem}\par\medskip
   \noindent  
   \textbf{Lemma \thetheorem.} \em \rmfamily}
\newenvironment{corollary}{\refstepcounter{theorem}\par\medskip
   \noindent  
   \textbf{Corollary \thetheorem.} \em \rmfamily}
\newenvironment{theorem}{\refstepcounter{theorem}\par\medskip
   \noindent  
   \textbf{Theorem \thetheorem.} \em \rmfamily}
\newenvironment{remark}[1]{\refstepcounter{theorem}\par\medskip
   \noindent  
   \textbf{Remark \thetheorem~(#1).} \em \rmfamily}
{\medskip}
\newenvironment{remarkPlain}{\refstepcounter{theorem}\par\medskip
   \noindent  
   \textbf{Remark \thetheorem.} \em \rmfamily}
{\medskip}

\newenvironment{proof}{%\par\medskip
   \noindent \textit{Proof:} \rmfamily}{\hfill $\square$\medskip}

\begin{document} 
\date{}

\title{\LARGE \bf Fictitious Play in Markov Games with Single Controller}

% Anonymized submission.
\author{Muhammed O. Sayin\thanks{M. O. Sayin is with the Department of Electrical and Electronics Engineering, Bilkent University, Ankara Turkey {\tt \small sayin@ee.bilkent.edu.tr}}
\and Kaiqing Zhang\footnotemark[2]
\and Asuman Ozdaglar \thanks{K. Zhang and A. Ozdaglar are with the Laboratory for Information and Decision Systems, Massachusetts Institute of Technology, Cambridge MA, United States of America {\tt \small kaiqing@mit.edu, asuman@mit.edu}}}
\maketitle
\thispagestyle{empty} 

\bigskip

% Abstract. Note that this must come before \maketitle.
\begin{center}
\textbf{Abstract}
\end{center}
Certain but important classes of strategic-form games, including zero-sum and identical-interest games, have the \textit{fictitious-play-property} (FPP), i.e., beliefs formed in fictitious play dynamics always converge to a Nash equilibrium (NE) in the repeated play of these games. Such convergence results are seen as a (behavioral) justification for the game-theoretical equilibrium analysis. Markov games (MGs), also known as stochastic games, generalize the repeated play of strategic-form games to dynamic multi-state settings with Markovian state transitions. In particular, MGs are standard models for multi-agent reinforcement learning -- a reviving research area in learning and games, and their game-theoretical equilibrium analyses have also been conducted extensively. However, whether certain classes of MGs have the FPP or not (i.e., whether there is a behavioral justification for equilibrium analysis or not) remains largely elusive. In this paper, we study a new variant of   fictitious play dynamics for MGs and show its convergence to an NE in $n$-player identical-interest MGs in which a single player controls the state transitions. Such games are of interest in communications, control, and economics applications. Our result together with the recent results in \citep{ref:Sayin20} establishes the FPP of two-player zero-sum MGs and $n$-player identical-interest MGs with a single controller (standing at two different ends of the MG spectrum from fully competitive to fully cooperative).
\bigskip


\bigskip

\bigskip

%%%%%%%
%{\ } %%%%%%%
\newpage 
\setcounter{page}{1} %%%%\linespread{.9}
%%\tableofcontents
%
%%\date{\today}
%%%%%
%%%%%
\begin{spacing}{1.245}

%%%%%%%%%%%%%%%%%%%%%%%%%%%%%%%%%
\section{Introduction}  

Markov games (MGs), also known as stochastic games, since their introduction in \citep{ref:Shapley53}, have been broadly used to model strategic interactions of multiple agents in dynamic environments with multiple states. The players' actions affect not only their immediate \textit{stage-payoffs}, but also the state transitions, and therefore, their future stage-payoffs.\footnote{Hereafter, we use {\it players}  and {\it agents} interchangeably.}  This powerful  framework to  model the sequential decision-making of multiple agents finds broad applications in both Engineering and Economics  \citep{neyman2003stochastic,bacsar1998dynamic}. Moreover, MGs also serve as  the fundamental  framework for multi-agent reinforcement learning \citep{ref:Littman94,busoniu2008comprehensive,zhang2019multi}. 

Nash equilibrium (NE) \citep{nash1951non}, on the other hand, has  been broadly used as a solution concept in game theory. One important justification of NE is that it is the natural outcome of  the myopic learning dynamics of players that  take greedy best response actions. Such a perspective has been extensively studied in strategic-form games (also known as normal-form or one-shot games), for best-response and fictitious-play types of learning  dynamics \citep{ref:Hofbauer02,ref:Leslie05,marden2009payoff,swenson2018best}.  In particular, these non-equilibrium adaptation dynamics are referred to as being {\it uncoupled/independent}, and these games where fictitious-play dynamics converge are referred to as having {\it fictitious-play-property} (FPP) \citep{ref:Robinson51,ref:Monderer96JET,ref:Miyasawa61,ref:Sela99}. For strategic-form games, it is well-known that several important classes of games enjoy the FPP, ranging from fully competitive to fully cooperative ones, with no modification of the fictitious play dynamics being used. This is especially a desired property for independent learning with uncoupled dynamics, where the players are oblivious to the structure of the underlying game while learning. 

In stark contrast,  the FPP of MGs remains largely  elusive. Limited results have been established on uncoupled learning dynamics of non-equilibrium adaptation for MGs, as well as using it as the  justifications  for the equilibrium therein. Recently, \cite{ref:Leslie20, ref:Sayin20,baudin2021best}  are the first set of results along this line, with focuses on either zero-sum or identical-interest MGs. Moreover, some learning dynamics \citep{ref:Leslie20,baudin2021best} are not fully independent  in that all the players track a common set of parameters. 
This naturally leads to the following open question we are interested in: 

\begin{center}
{\it Can we design independent learning dynamics  with uncoupled update rules, which enjoy the fictitious-play-property for more than one class of Markov games?}  
\end{center}

To shed light on this open problem, we study the same (synchronous- and model-based version) learning dynamic in \citep{ref:Sayin20}, an uncoupled fictitious-play dynamic that provably converges for zero-sum MGs, and investigate its convergence property in an important class of  games: identical-interest MGs with single-controller. We summarize our contributions as follows. 

\paragraph{Contributions.} We study \textit{two-timescale fictitious-play dynamics} for MGs, with independent and uncoupled update rules that combine the classical fictitious-play in the repeated play of strategic-form games with the $Q$-learning in solving Markov decision processes. We show that this natural learning dynamic  converges to an NE in both $n$-player identical-interest MGs (with single-controller) and two-player zero-sum MGs. In other words, these MGs, standing at two different ends of the MG spectrum, have the FPP. To the best of our knowledge, this appears to be the first fictitious-play type learning dynamics for MGs that enjoys this property. To establish the results, we develop new techniques to handle the challenges due to: 1) non-uniqueness of the NE value and the non-contracting property of the NE operator in identical-interest games; 2) non-monotonicity of the value function estimates when studying the discrete-time updates directly; 3) the deviation from the identical-interest structure of stage-games during learning, caused by the independent and local updates by each player.  
 
\subsection{Related work} 

We summarize the most related literature as follows. 

\paragraph{Fictitious-play dynamics/property.} Fictitious-play, a simple and independent learning dynamic that has been extensively studied for the repeated play of strategic-form games, was first introduced by \cite{brown1951iterative}. The dynamic has then been shown to converge to an equilibrium in multiple classes of strategic-form games, including zero-sum \citep{ref:Robinson51}, identical-interest  \citep{ref:Monderer96JET}, and certain general-sum games \citep{ref:Miyasawa61,ref:Sela99,ref:Berger05,ref:Berger08}. Recall that these games are referred to as having the FPP \citep{ref:Monderer96JET}. 

For MGs, the FPP has not been understood until recently in   \citep{ref:Leslie20,ref:Sayin20,baudin2021best}, which are the most related works to the present one. \cite{ref:Leslie20} presents  a continuous-time best-response dynamic for zero-sum MGs and embeds the discrete-time update into a continuous-time one. A single continuation payoff (common among the players) is maintained by all players, which makes the update rule not fully decoupled. \cite{ref:Sayin20} proposes fictitious play dynamics with uncoupled update rules, also for the zero-sum setting, where the continuation payoffs are updated locally using each player's own belief, yielding a more natural dynamic. Our learning dynamic is  thus also based on that in \citep{ref:Sayin20}. Very recently, \cite{baudin2021best} studies fictitious play for identical-interest MGs. The learning dynamic  also uses a common continuation payoff, and the discrete-time dynamic with convergence guarantees follows a {\it single-timescale} update rule. It is unclear if the same learning dynamic converges in other types of MGs. In fact, studying the convergence of {\it two-timescale} learning dynamics with local updates has been posted as an open question in \citep{baudin2021best}, which is one of the main focuses of the present work. 

\paragraph{Independent learning in MGs.} Besides the fictitious-play dynamics in \citep{ref:Leslie20,ref:Sayin20,baudin2021best}, other independent learning dynamics  have also been proposed for MGs. \cite{ref:Arslan17} studies decentralized $Q$-learning for MGs by focusing only on stationary pure strategies (saying which pure action to play at which state). This restriction allows them to transform the underlying MG into a strategic-form game in which actions correspond to stationary pure strategies (which are finitely many contrary to stationary mixed strategies). Players can learn the payoffs of the associated strategic-form game (without observing others' actions) with coordinated exploration phases in which they do not change their strategies to create a stationary environment. The dynamic presented can converge to a (stationary pure-strategy) equilibrium if the associated normal-form game is weakly acyclic with respect to best (or better) response dynamics. The finite-sample complexity of the algorithm is also established recently in \citep{gao2021finite}. In contrast, our learning dynamic  can converge to a stationary mixed-strategy equilibrium, which is essential for a global convergence result across the MG spectrum, as a pure-strategy equilibrium does not exist in general, e.g., in zero-sum games. \cite{perolat2018actor} develops actor-critic  learning dynamics that are decentralized, for a special class of MGs with a ``multistage''  structure, where  each state is assumed to be visited at most once. In \citep{daskalakis2020independent}, independent policy gradient methods with a two-timescale (asymmetric) stepsizes between players have been studied for the zero-sum setting, with non-asymptotic  convergence guarantees. Later, \cite{sayin2021decentralized} developed decentralized $Q$-learning dynamic that is symmetric, but with only asymptotic convergence guarantees in the zero-sum setting.  More recently, for Markov potential games, which also includes identical-interest MGs as an example, such independent policy gradient algorithms are also shown to converge  \citep{leonardos2021global,zhang2021gradient,fox2021independent,ding2022independent}. For episodic MGs, \citep{jin2021v,song2021can,mao2022provably} establish the regret guarantees of decentralized learning algorithms  in the online exploration setting. 
 
 \paragraph{MGs with single controller.} An important subclass of MGs is the ones with single controller \citep{filar2012competitive,parthasarathy1981orderfield}, where one of the players dominates and controls the transitions of the system dynamics (though the reward functions are still affected jointly by all players). Such a model finds applications in communications,  control, and economics  \citep{bacsar1986dynamic,eldosouky2016single}. It also has natural connection with  sequential (or online) learning \citep{cesa2006prediction,guan2016regret}. Learning in single-controller MGs are mostly focused on the zero-sum case \citep{brafman2000near,guan2016regret,qiu2021provably}. \cite{brafman2000near} studies a model-based approach with polynomial time complexity in  achieving near-optimal return. \cite{guan2016regret} investigates the relationship between regret minimization  and solving single-controller MGs, by reducing this model  to an online linear optimization problem. \cite{qiu2021provably} develops a policy optimization algorithm based on the idea of fictitious play, with regret guarantees in the episodic setting. It is unclear yet if these algorithms also converge to an NE in other classes of MGs.
 
%a new Lyapunov function and  

\subsection{Organization} The rest of the paper is organized as follows. We provide a formulation of MGs (with single controller) in \S\ref{sec:Markov} and describe the fictitious play dynamic in MGs in \S\ref{sec:FP}. We present the main convergence results  and the proof of convergence in \S\ref{sec:result}.
We conclude the paper in \S\ref{sec:conclusion} with some remarks. 

\section{Markov Games with Single Controller}\label{sec:Markov}

Consider an $n$-player MG described by a tuple $\langle S, A, \{r^i\}_{i\in[n]}, p, \gamma \rangle$.\footnote{For easy referral, we set player $i$ as the typical player while $-i:=\{j\in [n]\; |\; j\neq i\}$ corresponds to the set of players other than player $i$.}  The game has \textit{finitely} many states and $S$ denotes the set of states. At each state $s\in S$, each player $i$ can take an action $a^i$ from a \textit{finite} action set $A^i$, and $A=\bigtimes_i A^i$ denotes the set of action profiles $a=(a^i)_{i\in[n]}$.\footnote{The formulation can be extended to state-dependent action sets straightforwardly.} Over discrete-time $k=0,1,2,\ldots$, the state of the game, $s$, transitions to a state $s'$ according to the transition probability $p(s'|s,a)$ depending only on the current state $s$ and action profile $a$. At each stage $k$, each player $i$ receives a \textit{stage-payoff} $r^i(s,a)$ depending only on the current state $s$ and action profile $a$ while the players take actions simultaneously. Their objective is to maximize the discounted sum of their expected stage-payoffs over infinite horizon with the discount factor $\gamma\in [0,1)$.

MGs can be viewed as an extension of Markov decision processes to multi-agent settings. \cite{ref:Shapley53} (and later \cite{ref:Fink64}) showed that there always exists a \textit{Markov stationary} equilibrium in two-player zero-sum (and $n$-player general-sum) MGs such that players take actions according to {stationary} (possibly mixed) strategies depending only on the current state.\footnote{Such equilibrium is also referred to as {\it Markov perfect equilibrium}   \citep{maskin1988theory,maskin1988theoryb}.} We denote the stationary mixed-strategy of player $i$ by $\pi^i:S\rightarrow \Delta(A^i)$.\footnote{We denote the probability simplex over the set $A$ by $\Delta(A)$.} Correspondingly, $\pi=(\pi^i)_{i\in[n]}$ denotes the strategy profile and $\Pi$ denotes the space of strategy profiles, i.e., $\pi\in \Pi$. We define
\begin{equation}\label{eq:u}
u^i(s;\pi) := \mathbb{E}\left\{\sum_{k=0}^{\infty} \gamma^k r^i(s_k,a_k)\;\Big|\;s_0 = s\right\},\quad\forall s\in S\mbox{ and }\pi\in\Pi,
\end{equation}
where $(s_k,a_k)$ denotes the state and action profile at stage $k$, and the expectation is taken with respect to the randomness induced from the stochastic state transitions and mixed strategies of players. With slight abuse of notation, we let $u^i(\pi):= \mathbb{E}\{u^i(s_0,\pi)\}$, where the expectation is taken with respect to the initial state distribution. Therefore, $u^i(\pi)$ corresponds to the discounted sum of expected stage-payoffs of player $i$ under strategy profile $\pi$.

\begin{definition}[Markov Stationary Nash Equilibrium]\label{def:NE}
We say that strategy profile $\pi_*\in \Pi$ is a Markov stationary Nash equilibrium of the $n$-player MG provided that
\begin{equation}
u^i(\pi_*) \geq u^i(\pi^i,\pi_*^{-i}),\quad\forall \pi^i\mbox{ and } i =1,\ldots,n.
\end{equation} 
\end{definition}

Hereafter, NE refers  to Markov stationary Nash equilibrium. We say that an MG has \textit{zero-sum} or \textit{identical-interest} structure if $\sum_i r^i(s,a) = 0$ or $r^i(s,a) = r(s,a)$ for all $(s,a)$ for some $r:S\times A\rightarrow \mathbb{R}$, respectively. In this paper, we focus on \textit{single-controller} MGs where state transitions probabilities depend on the actions of a single player, e.g.,
\begin{equation}\label{eq:pp}
p(s'|s,a) = p(s'|s,a^i),\quad\forall (s,a,s').
\end{equation}
Note that since the reward functions, $r^i(s,a)$'s, are  affected by the joint action of all players, the accumulated expected payoff of player $i$ still depends on the joint strategy of all players. Hence, when the strategy of other players changes over time, the environment faced by one player is still {\it non-stationary}. This is  the key challenge in establishing the convergence of learning in MGs. 

Indeed, single-controller MGs are common models in the literature \citep{parthasarathy1981orderfield,filar2012competitive}, and find broad applications in communications \citep{eldosouky2016single} and traveling inspector problems \citep[Chapter 6]{filar2012competitive}. They also have natural connections with regret minimization for sequential (or online) learning \citep{guan2016regret,cesa2006prediction}. 

\section{Fictitious Play in Markov Games}\label{sec:FP}

Within an MG, stage-wise interactions among players can be viewed as they are playing \textit{auxiliary stage-games} specific to each state whenever the associated state gets visited. In each stage-game, players simultaneously take actions while they can mix their actions independently. Players observe the joint action of all players and receive the associated immediate stage-payoff.  However, the payoffs of these stage-games consist of immediate stage-payoffs and continuation payoffs (due to the objectives \eqref{eq:u} defined over infinite horizon). The players can compute the continuation payoff based on the observations they make. We focus on the question that whether {\it non-equilibrium adaptation} of learning agents can converge to a stationary (mixed-strategy) equilibrium of the underlying MG or not if they adopt learning dynamics similar to the ones studied for strategic-form games with repeated play, such as fictitious play and its variants.

Formally, if player $i$ knew that players $-i$ would play according to $\pi^{-i}$ starting from the next stage, player $i$'s payoff in the auxiliary stage-game associated with state $s$, denoted by $Q^i(s,a;\pi^{-i})$ and called \textit{Q-function}, would satisfy the following fixed-point equation
\begin{align}\label{eq:fixedpoint}
Q^i(s,a;\pi^{-i})=r^i(s,a)+\gamma\cdot \sum_{\tilde s}p(\tilde s|s,a)\max_{\tilde a^i\in A^i}\mathbb{E}_{\tilde a^{-i}\sim\pi^{-i}(s)}\big\{Q^i(\tilde s,\tilde a;\pi^{-i})\big\}\quad\forall (s,a).
\end{align}
This follows from the backward induction principle that player $i$ would always take the actions maximizing her expected utility in \eqref{eq:u}. 
Correspondingly, the value of state $s$, denoted by $v^i(s;\pi^{-i})$ and called \textit{value function}, would be given by 
\begin{equation}
v^i(s;\pi^{-i})=\max_{a^i\in A^i}\mathbb{E}_{a^{-i}\sim\pi^{-i}(s)}\big\{Q^i(s,a;\pi^{-i})\big\},\quad\forall s.
\end{equation}
Furthermore, if player $i$ also knew that players $-i$ would play according to $\pi^{-i}$ in the current auxiliary game, she would take the best response action, denoted by $a_*^i:S\rightarrow A^i$, satisfying
\begin{equation}
a_*^i(s) \in \argmax_{a^i \in A^i} \mathbb{E}_{a^{-i}\sim\pi^{-i}(s)} \big\{Q^i(s,a;\pi^{-i})\big\},\quad\forall s.
\end{equation}

Neither the opponent's strategy nor the $Q$-function are directly available to player $i$ in these auxiliary games. Therefore, each player $i$ can form beliefs on every other player's stationary (mixed) strategy and her (local) $Q$-function based on an erroneous assumption that they are stationary as in the classical fictitious play. Then, they can update these beliefs independently based on the observations they make within the underlying MG. For the ease of exposition, we consider that every player follow the same learning dynamic with the same learning rates (or step sizes) and initializations. Hence, all players $-i$ form the same belief on the stationary strategy of player $i$. We denote this belief by $\belief^i_k:S\rightarrow \Delta(A^i)$ at stage $k$. Similarly, we denote the belief of player $i$ on her $Q$-function at stage $k$ by $\Q^i_k:S\times A\rightarrow \mathbb{R}$. For notational convenience, we also introduce the value function estimates given by
\begin{equation}\label{eq:value}
\vfunc_k^i(s) := \max_{a^i\in A^i} \;\mathbb{E}_{a^{-i}\sim \belief_k^{-i}(s)} \{\Q_k^i(s,a^i,a^{-i})\}
\end{equation}
and the best response action given by
\begin{equation}\label{eq:action}
a_k^i(s) \in \argmax_{a^i\in A^i} \;\mathbb{E}_{a^{-i}\sim \belief_k^{-i}(s)} \{\Q_k^i(s,a^i,a^{-i})\}.
\end{equation}
Correspondingly, we have $\vfunc_k^i(s) = \mathbb{E}_{a^{-i}\sim \belief_k^{-i}(s)} \{\Q_k^i(s,a^i_k(s),a^{-i})\}$.

The players always take the best response \eqref{eq:action} according to the beliefs they form and they update their beliefs according to an update rule combining the classical fictitious play and $Q$-learning together. From player $i$'s viewpoint, the update rule is given by
\begin{subequations}\label{eq:updateFP}
\begin{align}
& \belief_{k+1}^{j}(s) = \belief_k^{j}(s) +  \alpha_{k}\Big(a_k^{j}(s) - \belief_k^{j}(s)\Big),\quad \forall j\neq i \mbox{ and } s\in S,\label{eq:updatepi}\\
& \Q_{k+1}^i(s,a) = \Q_k^i(s,a) + \beta_k\left(r^i(s,a) + \gamma \sum_{\tilde{s}}p(\tilde{s}|s,a) \vfunc_k^i(\tilde{s}) - \Q_k^i(s,a)\right),\quad\forall (s,a),\label{eq:updateQ}
\end{align}
\end{subequations}
where $\{\alpha_k,\beta_k\in (0,1)\}_{k\geq 0}$ are step sizes and the beliefs are initialized as, e.g., $\belief^{j}(s) = \frac{1}{|A^j|}\mathbf{1}$ and $\Q_0^i(s,a) = 0$ for all $(s,a)$.\footnote{Consider actions as pure strategies, i.e., $a^i\in A^i\subset \Delta(A^i)$.} In \eqref{eq:updatepi}, $\belief_k^j$ gets updated to a convex combination of the current action and the previous belief. On the other hand, in \eqref{eq:updateQ}, $\Q_k^i$ gets updated to a convex combination of the $Q$-function realized (according to one step iteration of the fixed-point equation \eqref{eq:fixedpoint} based on the value function estimate \eqref{eq:value}) and the previous belief. The weights of the new observations in these convex combinations are determined according to the step sizes $\{\alpha_k,\beta_k\}_{k\geq 0}$.

Note that if there was a single state, then the underlying MG would reduce to the repeated play of a strategic-form game, and correspondingly, \eqref{eq:updateFP} would reduce to \eqref{eq:updatepi} for which $\Q_k^i\equiv r^i$, which is indeed the classical fictitious play dynamic. On the other hand, if there was a single player, then the MG would reduce to a Markov decision process, and correspondingly, \eqref{eq:updateFP} would reduce to \eqref{eq:updateQ}, which is indeed the $Q$-value iteration with smoothing updates  (whose model-free version is known as $Q$-learning). The learning dynamic in \eqref{eq:updateFP} combines them together with different step sizes $\{\alpha_k,\beta_k\}_{k\geq 0}$ for learning in MGs.

\begin{remark}{Comparison to existing related  learning dynamics}
The learning dynamic in \eqref{eq:updateFP} is a synchronous and model-based version of the fictitious play dynamics in \citep{ref:Sayin20} focusing on learning in zero-sum MGs.\footnote{The update \eqref{eq:updateFP} is more like a computational method similar to the ones in \citep{ref:Leslie20,baudin2021best} contrary to \citep{ref:Sayin20} since players play the auxiliary stage-game associated with each state at every stage. This yields a relaxation on the convergence guarantees by not requiring the underlying Markov chain to ensure infinitely often visit at every state.} Two important features of the learning dynamic are: 1) the belief update and the Q-function update are performed in a two-timescale fashion; 2) each player maintains her own local estimates of the Q-functions, which are generally not common among players. In contrast, in the closely related recent  works \citep{ref:Leslie20,baudin2021best}, a single continuation
payoff (common among players) is assumed to be maintained during learning. This way, the stage-games encountered during learning are always zero-sum \citep{ref:Leslie20} or identical-interest \citep{baudin2021best}, and some {\it implicit}  coordination among the players is required. Our learning dynamic is coordination-free and completely uncoupled, and are thus believed to be more natural. In fact, studying  this two-timescale learning dynamic with independent $Q$-function updates has been posted as an interesting open question in \citep{baudin2021best}, with non-trivial technical challenges to address. Finally, another motivation of studying \eqref{eq:updateFP} is to find a unified learning dynamic that converges for both zero-sum and identical-interest MGs, i.e., being agnostic to the types of games, a desired property of uncoupled  dynamics. 
\end{remark}

\section{Convergence Results and Proofs}\label{sec:result}

Recall that the classical fictitious play is known to converge to an NE in certain but important classes of strategic-form games played repeatedly, such as zero-sum and identical-interest ones. The following theorem shows that the two-timescale fictitious-play dynamic in Eq.  \eqref{eq:updateFP} possesses similar universality by converging to an NE in both two-player zero-sum and multi-player identical-interest MGs with single controller. The proof is provided later in \S\ref{sec:proof}. 

\begin{theorem}\label{thm:main}
The update \eqref{eq:updateFP} converges to an NE, described in Definition \ref{def:NE}, in single-controller MGs with two-player zero-sum or multi-player identical-interest structure provided that the step sizes satisfy the usual two-timescale learning conditions that
\begin{itemize}
\item[(i)] Vanishing rates: $\alpha_k\rightarrow 0$ and $\beta_k\rightarrow 0$, as $k\rightarrow\infty$,
\item[(ii)] Sufficiently slow decay: $\sum_{k\geq 0} \alpha_k = \infty$ and $\sum_{k\geq 0} \beta_k = \infty$,
\item[(iii)] Sufficiently fast decay: $\sum_{k\geq 0} \alpha_k^2 < \infty$,  
\item[(iv)] Two-timescale rates: $\alpha_k\geq \beta_k$ for all $k\geq 0$ and $\beta_{k}/\alpha_k \rightarrow 0$, as $k\rightarrow\infty$.
\end{itemize}
Particularly, there exists $Q_*:S\times A\rightarrow \mathbb{R}$ such that
$$
\lim_{k\rightarrow\infty}\Q_k^i(s,a) = \Q_*(s,a)\quad\forall (i,s,a).
$$

Furthermore, in the zero-sum case, $Q_*$ corresponds to the $Q$-function associated with some stationary equilibrium $\pi_*=(\pi_*^i)_{i\in[n]}$ of the underlying game and
$$
\lim_{k\rightarrow\infty}\pi_k^i(s) = \pi_*^i(s),\quad\forall (i,s).
$$
On the other hand, in the identical-interest case, if the auxiliary stage-game of each state $s$ with the common payoff $\Q_*(s,\cdot)$ has finitely or countably many equilibria, then we also have that $Q_*$ corresponds to the $Q$-function associated with some stationary equilibrium $\pi_*=(\pi_*^i)_{i\in[n]}$ of the underlying game and
$$
\lim_{k\rightarrow\infty}\pi_k^i(s) = \pi_*^i(s),\quad\forall (i,s).
$$
\end{theorem}

We emphasize that the conditions listed are sufficient to ensure convergence of the update \eqref{eq:updateFP} in both classes of games. For example, the dynamic in \eqref{eq:updateFP} can converge to an equilibrium in two-player zero-sum MGs without Assumption $(iii)$ on sufficiently fast decay of $\alpha_k$. On the other hand, \eqref{eq:updateFP} can converge to an equilibrium in identical-interest MGs with single controller also in the single-timescale scheme where $\alpha_k=\beta_k$. Furthermore, the additional condition for the convergence of beliefs is standard in the analysis of {\it discrete-time} learning dynamics in potential/identical-interest  games \citep{heliou2017learning,ref:Benaim05,candogan2013dynamics}  (except the seminal result \citep{ref:Monderer96JET}).

The following corollary to Theorem \ref{thm:main} shows that the convergence result can be generalized to the case where the stage-payoffs satisfy the following potential-game-like condition similar to the case in one-shot games. The proof is deferred to Appendix \ref{app:potential}.

\begin{corollary}\label{cor:potential}
Suppose that the step sizes satisfy the conditions listed in Theorem \ref{thm:main} and the stage-payoff functions satisfy
\begin{equation}\label{eq:condition}
r^{j}(s,\tilde{a}^j,a^{-j}) - r^{j}(s,a) = r^i(s,\tilde{a}^j,a^{-j}) - r^i(s,a),\quad\forall (s,a), \tilde{a}^j,\mbox{ and }j\neq i
\end{equation}
given that player $i$ is the single controller. Then, the update \eqref{eq:updateFP} converges to an equilibrium in single-controller MGs. Particularly, there exists $Q_*^i:S\times A\rightarrow \mathbb{R}$ for each $i$, which is not necessarily common now, such that
$$
\lim_{k\rightarrow\infty}\Q_k^i(s,a) = \Q_*^i(s,a)\quad\forall (i,s,a).
$$
Furthermore, if the auxiliary stage-game of each state $s$ with the payoffs $\{\Q_*^i(s,\cdot)\}_{i\in[n]}$ has finitely or countably many equilibria, then we also have that $\{Q_*^i\}_{i\in [n]}$ correspond to the $Q$-functions associated with some stationary equilibrium $\pi_*=(\pi_*^i)_{i\in[n]}$ of the underlying game and
$$
\lim_{k\rightarrow\infty}\pi_k^i(s) = \pi_*^i(s),\quad\forall (i,s).
$$
\end{corollary}

Zero-sum games and identical-interest games stand at the two extreme ends of the game spectrum from fully competitive to fully cooperative. They possess distinct features. For example, the equilibrium value of a zero-sum (strategic-form) game is {\it unique} even though there may exist multiple equilibria. Furthermore, minimax value of a game is a non-expansive function like the maximum value. Therefore, \cite{ref:Shapley53} could introduce a contraction operator for two-player zero-sum MGs as a counterpart of the Bellman operator in Markov decision processes. Later \cite{ref:Leslie20,ref:Sayin20} showed that the contraction property in the evolution of the value function estimates could be approximated with asymptotically negligible error also in \textit{non-equilibrium} learning dynamics through a two-timescale learning scheme. 

\subsection{Proof of Theorem \ref{thm:main}}\label{sec:proof}

The proof of Theorem \ref{thm:main} for two-player zero-sum MGs (with single controller) follows from the identical steps in \citep[Theorem 4.3]{ref:Sayin20} where the convergence properties of the asynchronous version of \eqref{eq:updateFP} is characterized. On the other hand, equilibrium values of an identical-interest game are not necessarily unique. In the absence of powerful non-expansiveness and correspondingly contraction property, we need a different technical tool to characterize its convergence properties.

The main premise behind the proof for $n$-player identical-interest case is that the limiting differential inclusion of the \eqref{eq:updatepi} is the continuous-time best response dynamic in an identical-interest game due to the two-timescale framework. Therefore, the maximum expected values of auxiliary stage games are monotonically non-decreasing in this continuous-time approximation. Correspondingly, if the value function estimates were monotonically non-decreasing in the original discrete-time updates, then the $Q$-function estimates would also be monotonically non-decreasing. Hence, we could have concluded their convergence since they are bounded by the update rule \eqref{eq:updateQ}. However, the discrete-time dynamic does not necessarily lead to an increase in the value function estimates across subsequent stages in general. To address this challenge, we consider the deviation from the monotonicity across multiple stages rather than just subsequent ones, as in \citep{baudin2021best}. 

\begin{remark}{Challenge due to {\it Independent} $Q$-update}\label{rem:independent}
Similar to the zero-sum case in \citep{ref:Sayin20}, another challenge arises due to the deviation from the identical-interest structure in the auxiliary stage-games since players update their beliefs on the $Q$-function according to \eqref{eq:updateQ} via the maximum expected continuation payoff they believe they would get, as described in  \eqref{eq:value}. This challenge would not be observed if players had a {\it common} $Q$-function estimates, i.e., $\Q_k^i\equiv \Q_k^{}$ for all $i$ for some $\Q_k$. For example, \cite{baudin2021best} uses an update similar to 
\begin{equation}\label{eq:updateQ2}
\Q_{k+1}^i(s,a) = \Q_k^i(s,a) + \beta_k\left(r(s,a) + \gamma \sum_{s'}p(s'|s,a) \mathbb{E}_{a'\sim \pi_k(s')}\{\Q_k^i(s',a')\} - \Q_k^i(s,a)\right),
\end{equation}
for all $(i,s,a)$, rather than \eqref{eq:updateQ}.
Such an update guarantees that each auxiliary stage game has identical-interest structure
provided that $\Q_0^i=\Q_0^{}$ for all $i$. However, \eqref{eq:updateQ2} is still prone to deviation from the identical-interest structure without a common initialization. Computation of $\mathbb{E}_{a'\sim \pi_k(s')}\{\Q_k^i(s',a')\}$ by player $i$ also implies that player $i$ forms belief $\pi^i$ on her own strategy as if she is playing according to a stationary mixed-strategy even though she always takes (greedy) best response actions against her opponents. Therefore, the independent $Q$-update \eqref{eq:updateQ} contrary to the coupled one \eqref{eq:updateQ2} is a relatively more natural dynamic  for practical applications.  
The characterization of its convergence properties can provide a stronger justification for equilibrium analysis in MGs. To address this challenge, we focus on single-controller MGs, where the auxiliary stage games are \textit{strategically equivalent} to identical-interest games if the beliefs on $Q$-functions are initialized the same and become strategically equivalent to identical-interest games at a sufficiently fast rate if they do not have common initialization. 
\end{remark}

Given the update \eqref{eq:updateQ}, we define
\begin{equation}\label{eq:V}
\XiV_k^i(s,a) := r(s,a) + \gamma \sum_{s'\in S}p(s'|s,a) \vfunc_k^i(s') - \Q_k^i(s,a),\quad\forall (i,s,a)
\end{equation}
such that 
\begin{equation}
\Q_{k+1}^i(s,a)= \Q_k^i(s,a) + \beta_k \XiV_k^i(s,a)\quad\forall (i,s,a).
\end{equation}
Note that $\Q_k^i(s,a)$, for each $(i,s,a)$, is bounded from above by $\frac{1}{1-\gamma}\max_{(s,a)}r(s,a)$ by the definition of the update \eqref{eq:updateQ} and \eqref{eq:value} since the step size $\{\beta_k\in (0,1)\}_{k\geq 0}$. 

The following proposition provides a monotonicity-like condition on the changes of the estimates accumulated across multiple stages (not just the subsequent ones) to prove the convergence of $\{\Q_k^i(s,a)\}_{k\geq 0}$. The proof is deferred to Appendix \ref{app:step0}. 

\begin{proposition}\label{prop:step0}
Consider a (real-valued) bounded sequence $\{\Q_k^i(s,a)\}_{k\geq 0}$ for each $(i,s,a)\in [n]\times S\times A$ (with $n<\infty$ and $|S\times A|<\infty$) evolving according to
\begin{equation}
\Q_{k+1}^i(s,a)= \Q_k^i(s,a) + \beta_k \XiV_k^i(s,a),\quad\forall (i,s,a)
\end{equation}
for some $\{\XiV_k^i(s,a)\}_{k\geq 0}$ for each $(i,s,a)$ and step size $\beta_k\geq 0$. If we have
\begin{equation}\label{eq:bound}
\liminf_{k_1\rightarrow\infty}\inf_{k_2\geq k_1} \sum_{k=k_1}^{k_2} \beta_k \XiV_k^i(s,a)\geq 0,\quad\forall (i,s,a),
\end{equation}
then there exists $Q_*^i:S\times A\rightarrow\mathbb{R}$ such that
\begin{equation}
\lim_{k\rightarrow\infty} Q_k^i(s,a) = Q_*^i(s,a),\quad\forall (i,s,a).
\end{equation}
\end{proposition}

Henceforth, we focus on proving \eqref{eq:bound} (through a more tractable lower bound) to conclude the convergence of \eqref{eq:updateQ}. To this end, we define an auxiliary parameter bounding $\XiV_k^i(s,a)$ for each $(s,a)$ from below as
\begin{equation}\label{eq:uu}
\uu_k^i := \min_{(s,a)}\left\{r(s,a) + \gamma \sum_{s'\in S} p(s'|s,a) \mathbb{E}_{a'\sim \pi_k(s)}\{\Q_k^i(s',a')\} - \Q_k^i(s,a)\right\},
\end{equation}
i.e., we have $\uu_k^i\leq \XiV_k^i(s,a)$ for all $(s,a)$, since
$\vfunc_k^i(s') - \mathbb{E}_{a'\sim \pi_k(s)}\{\Q_k^i(s',a')\} \geq 0$ for each $s'$ by the definition of $\vfunc_k^i$, as described in \eqref{eq:value}. Instead of \eqref{eq:bound}, we can, now, focus on proving the asymptotic non-negativity of the more tractable lower bound: 
\begin{equation}\label{eq:bound2}
\boxed{\liminf_{k_1\rightarrow\infty}\inf_{k_2\geq k_1} \sum_{k=k_1}^{k_2} \beta_k \uu_k^i \geq 0}.
\end{equation}

If we could have shown that there exists some $\kappa$ such that $\uu_k^i\geq 0$ for all $k\geq \kappa$, then we would have concluded \eqref{eq:bound2}. We do not necessarily have it. On the other hand, showing the asymptotic non-negativity of $\uu_k^i$ would not be sufficient to conclude \eqref{eq:bound2}. Hence, we look for some stronger conditions. The following lemma provides a characterization of the evolution of $\{\uu_k^i\}_{i\geq 0}$ (from below) in terms of some \textit{absolutely summable} sequence. The proof is deferred to Appendix \ref{app:stepi}.

\begin{lemma}\label{lem:stepi}
Given that players follow the dynamic  described in \eqref{eq:updateFP} in an identical-interest MG with single controller $i$, the evolution of $\uu_k^i$, described in \eqref{eq:uu}, 
satisfies the following inequality:
\begin{equation}\label{eq:seq}
\boxed{\uu_{k+1}^i \geq \uu_k^i (1-(1-\gamma)\beta_k) + \ue_k}
\end{equation} 
for all $k\geq 0$ and for some \textit{absolutely summable} sequence $\{\ue_k\}_{k\geq 0}$.
\end{lemma} 

\begin{remarkPlain}
We emphasize that the single-controller identical-interest structure plays an important role in ensuring that there exists such an absolutely summable sequence.
\end{remarkPlain}

Given that $\{\ue_k\}_{k\geq 0}$ are absolutely summable, we can next invoke the following lemma showing that $\{\uu_k^i\}_{k\geq 0}$ satisfying \eqref{eq:seq} also satisfies \eqref{eq:bound2}. Hence, it can also be of interest on its own. The proof is deferred to Appendix \ref{app:stepii}.

%%%%%%%%%%%%%%%%%
\begin{lemma}\label{lem:stepii}
Given step sizes $\{\beta_k\in(0,1)\}_{k\geq 0}$ vanishing, i.e., $\beta_k\rightarrow 0$ as $k\rightarrow\infty$, sufficiently slowly such that $\sum_{k\geq 0}\beta_k = \infty$, consider a sequence $\{\uu_k^i\}_{k\geq 0}$ satisfying \eqref{eq:seq}
for some discount factor $\gamma \in[0,1)$ and some absolutely summable error term $\ue_k$, i.e., $\sum_{k\geq 0}|\ue_k| <\infty$. Then, we have \eqref{eq:bound2}.
\end{lemma}
%%%%%%%%%%%%%%%

Lemmas \ref{lem:stepi} and \ref{lem:stepii} imply \eqref{eq:bound2}, and therefore, \eqref{eq:bound} by the definition of $\uu
_k^i$, as described in \eqref{eq:uu}. Hence, Proposition \ref{prop:step0} yields that $\{\Q_k^i\}_{i\in [n]}$, i.e., \eqref{eq:updateQ}, is convergent. In other words, there exists some $Q_*^i:S\times A\rightarrow \mathbb{R}$ such that
\begin{equation}\label{eq:limQ}
\lim_{k\rightarrow\infty} \Q_k^i(s,a) = Q_*^i(s,a).
\end{equation}
On the other hand, for players other than the controller, say player $j\neq i$, the payoff in the auxiliary stage game $\Q_k^j(s,a)$ is always strategically equivalent to $\Q_k^i(s,a)$, as shown in the proof of Lemma \ref{lem:stepi}. Note that we have not characterized the limit of $\Q_k^i$ yet. 

Lastly, we can conclude that the update \eqref{eq:updateFP} indeed converges to an equilibrium based on the following lemma (which can be viewed as a corollary to \citep[Theorem 4]{ref:Leslie06}). This lemma characterizes the limit set of the fictitious-play in terms of its limiting continuous-time best response dynamic for the cases where the underlying game becomes stationary asymptotically. Hence, it can also be of interest on its own. The proof is deferred to Appendix \ref{app:stepiii}.

%%%%%%%%%%%%%%%
\begin{lemma}\label{lem:stepiii}
Given step sizes $\{\alpha_k\in(0,1)\}_{k\geq 0}$ vanishing, i.e., $\alpha_k\rightarrow 0$ as $k\rightarrow\infty$, sufficiently slowly such that $\sum_{k\geq 0}\alpha_k = \infty$, consider the update of $\pi^i\in\Delta(A^i)$ for each $i\in [n]$ and finite set $A^i$, given by
\begin{equation}\label{eq:updatenew}
\pi_{k+1}^i = \pi_k^i + \alpha_k \left(a_k^i-\pi_k^i\right),\quad\forall i,
\end{equation}
where $a_k^i\in \Delta(A^i)$ satisfies
\begin{equation}
a_k^i \in \argmax_{a\in A} \mathbb{E}_{a^{-i}\sim \pi_k^{-i}}\{Q_k^i(a^i,a^{-i})\}.
\end{equation}
Suppose that $\Q_k^i(a) \rightarrow Q_*^i(a)$ for all $(i,a)$ as $k\rightarrow\infty$ for some  $Q_*^i:A\rightarrow\mathbb{R}$. Then, the limit set of \eqref{eq:updatenew} is a connected internally chain-recurrent set of the following best response differential inclusion
\begin{equation}
\dot{\pi}^i + \pi^i \in \argmax_{a^i\in A^i}\mathbb{E}_{a^{-i}\sim\pi^{-i}}\{\Q_*^i(a^i,a^{-i})\}.
\end{equation} 
\end{lemma}  
%%%%%%%%%%%%%%%

Based on \eqref{eq:limQ}, Lemma \ref{lem:stepiii} yields that the limit set of \eqref{eq:updatepi} is contained in the connected internally chain-recurrent set of the differential inclusion
\begin{equation}\label{eq:ode}
\dot{\pi}^i(s) + \pi^i(s) \in \argmax_{a^i\in A^i}\mathbb{E}_{a^{-i}\sim\pi^{-i}(s)}\{\Q_*^i(s,a^i,a^{-i})\},
\end{equation}  
which is the continuous-time best response dynamic in an identical-interest game with the payoff $Q_*^i(s,\cdot)$.
\citep[Theorem 5.5]{ref:Benaim05} yields that the limit set of every solution of \eqref{eq:ode} is a connected set of equilibria along which $\mathbb{E}_{a\sim\pi(s)}\{\Q_*^i(s,a)\}$ is constant. Hence, the limit set of \eqref{eq:updatepi} is a connected subset of equilibria along which $\mathbb{E}_{a\sim\pi_k(s)}\{\Q_*^i(s,a)\}$ is constant provided that the auxiliary game with the common payoff $\Q_*^i(s,\cdot)$ has finitely or countably many equilibria (as assumed in the theorem statement). In other words, the beliefs $\{\belief_k^i(s)\}_{i\in [n]}$ converge to one of these isolated equilibria, for each $s$. Given the convergence of these beliefs, the $Q$-function estimates of every player also converge to the $Q$-function associated with the equilibrium strategies and
\begin{equation}
\lim_{k\rightarrow\infty}\Q_k^j(s,a) = \Q_*^i(s,a),\quad \forall (s,a)\mbox{ and } j\neq i.
\end{equation}
This completes the proof. \hfill $\square$

\section{Discussions and Conclusions} \label{sec:conclusion}

In this paper, we investigated the convergence properties of a new variant of fictitious play dynamics  for $n$-player identical-interest MGs with single controller. Together with the fact that the same learning dynamic also converges to an equilibrium in two-player zero-sum MGs, we established, to the best of our knowledge, the first universal-type fictitious-play-property for more than one class of MGs. The results have thus further justified (Markov stationary) NE in MGs as an outcome of myopic non-equilibrium adaptation. We believe our results have opened up fruitful research directions for future work. 

\begin{itemize}
	\item \textbf{Fictitious-play-property for other classes of MGs.} Our results illustrate the promise of our two-timescale fictitious-play dynamic in achieving universal-type convergence in more than one class of MGs. It is interesting to  further expand the types of MGs that enjoys the fictitious-play-property, mirroring the results for  strategic-form games (cf. \citep{ref:Miyasawa61,ref:Sela99,ref:Berger05,ref:Berger08}). 
	\item \textbf{Model-free learning with asynchronous updates.} With a focus on the uncoupled  learning dynamics with independent $Q$-updates, we studied the synchronous-update rule with the knowledge of the transition dynamics. As a standard model for multi-agent reinforcement learning, it is imperative to investigate the convergence of our dynamics in the model-free asynchronous setting. Note that with common $Q$-updates and single-timescale update-rule, \citep{baudin2021best} has studied the asynchronous update case with small enough discount factor $\gamma$. The model-free learning for fictitious-play in MGs beyond the zero-sum case remains largely open. 
	\item \textbf{Convergence rate characterization \& faster rates.} It is known that in the worse-case, fictitious play can have exponentially-slow rate when learning in strategic-form games \citep{daskalakis2014counter}. It would be interesting to understand  and compare the convergence rates of the two-timescale fictitious-play in our work and \citep{ref:Sayin20}, with that of the single-timescale one  in \citep{baudin2021best}.  It is also worth exploring the effectiveness of regularization to accelerate convergence of fictitious play dynamics, as in strategic-form games \citep{cesa2006prediction}.  
\end{itemize}

%%%
%%% The acknowledgments section is defined using the "acks" environment
%%% (and NOT an unnumbered section). This ensures the proper
%%% identification of the section in the article metadata, and the
%%% consistent spelling of the heading.
\section*{Acknowledgments}
M. O. Sayin was supported by TUBITAK BIDEB 2232-B International Fellowship for Early Stage Researchers grant 121C124. K. Zhang and A. Ozdaglar were supported by DSTA grant 031017-00016. 

\appendix

\section{Proof of Proposition \ref{prop:step0}}\label{app:step0}

Note that if $\XiV_k^i(s,a)\geq 0$, then $\{\Q_k^i(s,a)\}_{k\geq 0}$ would form a non-decreasing bounded sequence, which implies the existence of its limit. However, we do not necessarily have $\XiV_k^i(s,a)\geq 0$. As in \citep{baudin2021best}, we can check monotonicity across multiple stages (from $k_1$ to $k_2+1$). For example, we have
\begin{equation}\label{eq:both}
\Q_{k_2+1}^i(s,a) - \Q_{k_1}^i(s,a) = \sum_{k=k_1}^{k_2} \beta_k \XiV_k^i(s,a),\quad\forall (i,s,a),
\end{equation}
where the right-hand side is still not necessarily non-negative. 
Showing the difference goes to zero as $k_1\rightarrow\infty$ would imply that $\{\Q_k^i(s,a)\}_{k\geq 0}$ forms a Cauchy sequence, and therefore, it is convergent in the underlying Banach space. 
A relatively mild alternative (aligned with the intuition on monotonicity) is to show that the right-hand side becomes non-negative asymptotically as $k_1\rightarrow\infty$ for any $k_2\geq k_1$, i.e.,
\begin{equation}\label{eq:boundapp}
\liminf_{k_1\rightarrow\infty}\inf_{k_2\geq k_1} \sum_{k=k_1}^{k_2} \beta_k \XiV_k^i(s,a)\geq 0\quad\Rightarrow\quad\liminf_{k_1\rightarrow\infty}\left(\inf_{k_2\geq k_1}\Q_{k_2+1}^i(s,a) - \Q_{k_1}^i(s,a)\right)\geq 0,
\end{equation}
for all $(i,s,a)$. Then, \eqref{eq:boundapp} yields that for any $\epsilon >0$, there exists $\kappa$ such that
\begin{equation}
\Q_{k}^i(s,a) \geq \Q_{\kappa}^i(s,a) - \epsilon, \quad \forall k\geq \kappa,
\end{equation}
for each $(i,s,a)$.\footnote{We can have a uniform $\epsilon>0$ since there are only finitely many $(i,s,a)$ triples.}  
This completes the proof due to the boundedness of the estimates.

\section{Proof of Lemma \ref{lem:stepi}}\label{app:stepi}

For the ease of notation, we define 
\begin{equation}\label{eq:U}
\XiU_k^i(s,a) := r(s,a) + \gamma \sum_{s'\in S} p(s'|s,a) \ufunc_k^i(s') - \Q_k^i(s,a),
\end{equation}
where $\ufunc_k^i(s') := \mathbb{E}_{a'\sim \belief_k(s')}\{\Q_k^i(s',a')\}$. Then, we have $\uu_k^i = \min_{(s,a)}\{\XiU_k^i(s,a)\}$ for all $i$. Due to this dependence, it is instructive to examine the evolution of $\XiU_k^i(s,a)$, given by
\begin{align}
\XiU_{k+1}^i(s,a) - \XiU_k^i(s,a) &= \gamma \sum_{s'\in S} p(s'|s,a) (\ufunc_{k+1}^i(s') - \ufunc_k^i(s')) - (\Q_{k+1}^i(s,a)-\Q_k^i(s,a))\nn\\
&= \gamma \sum_{s'\in S} p(s'|s,a) (\ufunc_{k+1}^i(s') - \ufunc_k^i(s')) - \beta_k \XiV_k^i(s,a)\nn\\
&= \gamma \sum_{s'\in S} p(s'|s,a) (\ufunc_{k+1}^i(s') - \ufunc_k^i(s')) - \beta_k \left(\XiU_k^i(s,a) + \gamma\sum_{s'\in S}p(s'|s,a)\Delta_k^i(s')\right)\nn\\
&= \gamma \sum_{s'\in S} p(s'|s,a) (\ufunc_{k+1}^i(s') - \ufunc_k^i(s') - \beta_k \Delta_k^i(s')) - \beta_k \XiU_k^i(s,a),\label{eq:Uiter}
\end{align}
where $\Delta_k^i(s'):= \vfunc_k^i(s')-\ufunc_k^i(s')\geq 0$ for each $s'$. By the definition of $\ufunc_k^i$, the difference term in the parenthesis can be written as
\begin{align}
\ufunc_{k+1}^i(s') - \ufunc_k^i(s') &= \mathbb{E}_{a\sim\pi_{k+1}(s')}\{\Q_{k+1}^i(s',a)\} - \mathbb{E}_{a\sim\pi_{k}(s')}\{\Q_{k}^i(s',a)\} \nn\\
&\stackrel{(a)}{=} \mathbb{E}_{a\sim\pi_{k+1}(s')}\{\Q_k^i(s',a)\} - \mathbb{E}_{a\sim\pi_{k}(s')}\{\Q_k^i(s',a)\} + \beta_k \mathbb{E}_{a\sim\pi_{k+1}(s')}\{\XiV_k^i(s',a)\}\nn\\
&\stackrel{(b)}{=} \alpha_k\left(\Delta_k^i(s') + \sum_{i\neq j} \Gamma_k^{ij}(s')\right) + O(\alpha_k^2) + \beta_k \mathbb{E}_{a\sim\pi_{k+1}(s')}\{\XiV_k^i(s',a)\},\label{eq:uiter}
\end{align} 
where $(a)$ follows from the update of $\Q_k^i$, as described in \eqref{eq:updateQ}, $(b)$ follows from the update of $\belief_k^j$ for each $j\in [n]$, as described in \eqref{eq:updatepi}, and we define 
\begin{equation}\label{eq:Gamma}
\Gamma_k^{ij}(s') := \mathbb{E}_{a\sim \pi_k(s')} \{\Q_k^i(s',a_k^j(s'),a^{-j}) - \Q_k^i(s',a)\},\quad \forall j\neq i.
\end{equation}
Note that $\uu_k^i \leq  \XiU_k^i(s,a) \leq \XiV_k^i(s,a)\}$ for each $(s,a)$. Therefore, combining \eqref{eq:Uiter} and \eqref{eq:uiter}, we obtain
\begin{align}\label{eq:iter}
\XiU_{k+1}^i(s,a) \geq \uu_k^i(1-(1-\gamma)\beta_k) + \gamma\sum_{s'\in S}p(s'|s,a) e_k^i(s'),
\end{align}
where we define
\begin{equation}\label{eq:ee}
\boxed{e_k^i(s') := \alpha_k\left(\left(1-\frac{\beta_k}{\alpha_k}\right)\Delta_k^i(s') + \sum_{i\neq j} \Gamma_k^{ij}(s')\right) + O(\stepA_k(s')^2)}.
\end{equation}
If we can find an absolutely summable lower bound on $e_k^i(\cdot)$, then the inequality \eqref{eq:iter} yields \eqref{eq:seq}.

Next, we formulate an absolutely summable lower bound on $e_k^i(\cdot)$ based on the single-controller property of the underlying MG. Since $a_k^i$, as described in \eqref{eq:action}, is a best response action, we can write $\Delta_k^i(s')=\vfunc_k^i(s')-\ufunc_k^i(s')$ also as
\begin{equation}\label{eq:Delta}
\Delta_k^{i}(s') = \mathbb{E}_{a\sim \pi_k(s')} \{\Q_k^i(s',a_k^i(s'),a^{-i}) - \Q_k^i(s',a)\}.
\end{equation}
We highlight the differences between $\Delta_k^i(s)\geq 0$, as described in \eqref{eq:Delta}, and $\Gamma_k^{ij}(s')$, as described in \eqref{eq:Gamma}. In particular, $a_k^j$ is a best response of player $j$ according to her payoff function $\Q_k^j$ in the associated auxiliary stage-game and her belief $\pi_k^{-j}$ about her opponents' strategies. Therefore, $\Gamma_k^{ij}$ is not necessarily non-negative quite contrary to $\Delta_k^i\geq 0$ if we do not have $\Q_k^i \equiv \Q_k^j$ (e.g., see Remark \ref{rem:independent}). 

On the other hand, by \eqref{eq:Gamma} and \eqref{eq:Delta}, the term $\Gamma_k^{ij}$ can be written as
\begin{equation}
\Gamma_k^{ij}(s) = \Delta_k^j(s) + \mathbb{E}_{a\sim\pi_k(s)}\{\deltaQ_k^{ij}(s,a_k^j(s),a^{-j}) - \deltaQ_k^{ij}(s,a)\},
\end{equation}
where the first term on the right hand side is non-negative and we define $\deltaQ_k^{ij}(s,a) := \Q_k^i(s,a) - \Q_k^j(s,a)$ for all $(s,a)$. We can show that 
\begin{align}\label{eq:gam}
\mathbb{E}_{a\sim\pi_k(s)}\{\deltaQ_k^{ij}(s,a_k^j(s),a^{-j}) - &\deltaQ_k^{ij}(s,a)\} = 0 \quad\Rightarrow\quad \Gamma_k^{ij} \equiv \Delta_k^j\geq 0,\quad\forall j\neq i.
\end{align}
in MGs with single controllers. Particularly, the update \eqref{eq:updateQ} yields that\footnote{We use the convention that $\prod_{m=l}^k c_m = 1$ if $k<l$.}
\begin{align}
\Q_{k+1}^i(s,a) = r(s,a)\sum_{l=0}^k \beta_l \left(\prod_{m=l+1}^k(1-\beta_m)\right)+\gamma\sum_{s'\in S} p(s'|s,a) \sum_{l=0}^k \vfunc_l^i(s') \beta_l \left(\prod_{m=l+1}^k(1-\beta_m)\right),\label{eq:QQ}
\end{align}
for all $k\geq 0$,  since the beliefs are initialized by $\Q_0^i(s,a)$ for all $(i,s,a)$. This implies that
\begin{align}\label{eq:delQQ}
\deltaQ_k^{ij}(s,a) = \gamma\sum_{s'\in S} p(s'|s,a) \sum_{l=0}^{k-1} (\vfunc_l^i(s')-\vfunc_l^j(s'))\beta_l \left(\prod_{m=l+1}^{k-1}(1-\beta_m)\right).
\end{align}
Recall that if player $i$ is the single controller, then we have 
\begin{equation}\label{eq:pp}
p(s'|s,\tilde{a}^j,a^{-j}) - p(s'|s,a) = p(s'|s,a^i)-p(s'|s,a^i) = 0,\quad\forall a,\tilde{a}^j\mbox{ and } j\neq i.
\end{equation}
Hence, \eqref{eq:delQQ} and \eqref{eq:pp} yield \eqref{eq:gam}.
Correspondingly, we have
\begin{equation}\label{eq:eee}
\gamma\sum_{s'}p(s'|s,a)e_k^i(s') \geq O(\alpha_k^2)=:\ue_k.
\end{equation}
Note that $\{\alpha_k^2\}_{k\geq 0}$ is absolutely summable by Assumption $(iii)$ listed in Theorem \ref{thm:main}.  Hence,  \eqref{eq:iter} and \eqref{eq:eee} lead to \eqref{eq:seq}.
This completes the proof. \hfill $\square$

\section{Proof of Lemma \ref{lem:stepii}} \label{app:stepii}
Given \eqref{eq:seq}, we can formulate a lower bound on $\uu_k^i$ in terms of $\uu_0^i$ and $\{\ue_k\}$:  
\begin{align}
\uu_{k}^i &\geq \uu_{k-1}^i(1-\tilde{\beta}_{k-1}) + \ue_{k-1}\nn\\
&\geq \uu_{k-2}^i(1-\tilde{\beta}_{k-2})(1-\tilde{\beta}_{k-1}) + \ue_{k-2}(1-\tilde{\beta}_{k-1})+\ue_{k-1}\nn\\
&\ldots\nn\\
&\geq \uu_0^i \prod_{l=0}^{k-1}(1-\tilde{\beta}_l) + \sum_{l=0}^{k-1}\ue_l\prod_{m=l+1}^{k-1}(1-\tilde{\beta}_m),\label{eq:lower}
\end{align}
where $\tilde{\beta}_k := (1-\gamma)\beta_k$ for notational convenience. We can incorporate the lower bound \eqref{eq:lower} on $\uu_k^i$ into the summation in \eqref{eq:bound2} as
\begin{align}
\sum_{k=k_1}^{k_2}\beta_k \uu_k^i &\geq \sum_{k=k_1}^{k_2} \beta_k\left(\uu_0^i \prod_{l=0}^{k-1}(1-\tilde{\beta}_l) + \sum_{l=0}^{k-1}\ue_l\prod_{m=l+1}^{k-1}(1-\tilde{\beta}_m)\right)\nn\\
&\geq -\sum_{k=k_1}^{k_2} \beta_k\left(|\uu_0^i| \prod_{l=0}^{k-1}(1-\tilde{\beta}_l) + \sum_{l=0}^{k-1}|\ue_l|\prod_{m=l+1}^{k-1}(1-\tilde{\beta}_m)\right)\nn \\
&= -\frac{|\uu_0^i|}{1-\gamma}\sum_{k=k_1}^{k_2} \tilde{\beta}_k\prod_{l=0}^{k-1}(1-\tilde{\beta}_l) - \frac{1}{1-\gamma} \sum_{k=k_1}^{k_2} \sum_{l=0}^{k-1}|\ue_l|\tilde{\beta}_k\prod_{m=l+1}^{k-1}(1-\tilde{\beta}_m).\label{eq:tilde}
\end{align}
By changing the order of summation at the second term, we have:
\begin{align}
- \frac{1}{1-\gamma}\sum_{k=k_1}^{k_2} \sum_{l=0}^{k-1}|\ue_l|\tilde{\beta}_k\prod_{m=l+1}^{k-1}(1-\tilde{\beta}_m) = &\;- \frac{1}{1-\gamma}\sum_{l=0}^{k_1-2} |\ue_l| \sum_{k=k_1}^{k_2} \tilde{\beta}_k\prod_{m=l+1}^{k-1}(1-\tilde{\beta}_m) \nn\\
&- \frac{1}{1-\gamma} \sum_{l=k_1-1}^{k_2-1}|\ue_l| \sum_{k=l+1}^{k_2} \tilde{\beta}_k\prod_{m=l+1}^{k-1}(1-\tilde{\beta}_m).\label{eq:tilde2}
\end{align}
We are interested in proving \eqref{eq:bound2} and
\begin{align}
\liminf_{k_1\rightarrow\infty}\inf_{k_2\geq k_1}\sum_{k=k_1}^{k_2}\beta_k \uu_k^i \geq&\; \liminf_{k_1\rightarrow\infty}\inf_{k_2\geq k_1}\left( -\frac{|\uu_0^i|}{1-\gamma}\sum_{k=k_1}^{k_2} \tilde{\beta}_k\prod_{l=0}^{k-1}(1-\tilde{\beta}_l)\right)\nn\\
&+\liminf_{k_1\rightarrow\infty}\inf_{k_2\geq k_1}\left(- \frac{1}{1-\gamma}\sum_{l=0}^{k_1-2} |\ue_l| \sum_{k=k_1}^{k_2} \tilde{\beta}_k\prod_{m=l+1}^{k-1}(1-\tilde{\beta}_m) \right)\nn\\
&+\liminf_{k_1\rightarrow\infty}\inf_{k_2\geq k_1}\left(- \frac{1}{1-\gamma} \sum_{l=k_1-1}^{k_2-1}|\ue_l| \sum_{k=l+1}^{k_2} \tilde{\beta}_k\prod_{m=l+1}^{k-1}(1-\tilde{\beta}_m)\right).\label{eq:bound3}
\end{align}
Therefore, showing the non-negativity of each term at the right-hand side implies \eqref{eq:bound2}. To this end, the following lemma enables us to rewrite the inner summations in \eqref{eq:tilde} and \eqref{eq:tilde2} as a difference of two partial products.

\begin{lemma} \label{lem:aux}
We have
\begin{align}
\sum_{k=k_1}^{k_2} \beta_k \prod_{l=k_0}^{k-1}(1-\beta_l) = \prod_{l=k_0}^{k_1-1}(1-\beta_l) - \prod_{l=k_0}^{k_2}(1-\beta_l)
\end{align}
and
\begin{align}
\sum_{k=k_1}^{k_2} \beta_k \prod_{l=k+1}^{k_0}(1-\beta_l) = \prod_{l=k_1+1}^{k_0}(1-\beta_l) - \prod_{l=k_2}^{k_0}(1-\beta_l).
\end{align}
\end{lemma}

\begin{proof}
By adding and subtracting one to the term $\beta_k$, we obtain
\begin{align}
\sum_{k=k_1}^{k_2} \beta_k \prod_{l=k_0}^{k-1}(1-\beta_l) &= \sum_{k=k_1}^{k_2} (1 - (1-\beta_k)) \prod_{l=k_0}^{k-1}(1-\beta_l)\nn\\
&= \sum_{k=k_1}^{k_2} \left(\prod_{l=k_0}^{k-1}(1-\beta_l) - \prod_{l=k_0}^{k}(1-\beta_l)\right)\nn\\
&= \prod_{l=k_0}^{k_1-1}(1-\beta_l) - \prod_{l=k_0}^{k_2}(1-\beta_l),\label{eq:tele1}
\end{align}
and
\begin{align}
\sum_{k=k_1}^{k_2} \beta_k \prod_{l=k+1}^{k_0}(1-\beta_l) &= \sum_{k=k_1}^{k_2} (1 - (1-\beta_k)) \prod_{l=k+1}^{k_0}(1-\beta_l)\nn\\
&= \sum_{k=k_1}^{k_2} \left(\prod_{l=k+1}^{k_0}(1-\beta_l) - \prod_{l=k}^{k_0}(1-\beta_l)\right)\nn\\
&= \prod_{l=k_1+1}^{k_0}(1-\beta_l) - \prod_{l=k_2}^{k_0}(1-\beta_l),\label{eq:tele2}
\end{align}
where \eqref{eq:tele1} 
and \eqref{eq:tele2} 
follow from telescoping the series.
\end{proof}

Based on Lemma \ref{lem:aux}, the first term on the right-hand side of \eqref{eq:bound3} is non-negative because the summation is bounded from below by
\begin{align}
-\frac{|\uu_0^i|}{1-\gamma}\sum_{k=k_1}^{k_2} \tbeta_k\prod_{l=0}^{k-1}(1-\tbeta_l) &= -\frac{|\uu_0^i|}{1-\gamma} \left(\prod_{l=0}^{k_1-1}(1-\tbeta_l) - \prod_{l=0}^{k_2}(1-\tbeta_l)\right)\nn\\
&\geq -\frac{|\uu_0^i|}{1-\gamma} \prod_{l=0}^{k_1-1}(1-\tbeta_l),
\end{align}
which does not depend on $k_2$ and goes to zero as $k_1\rightarrow\infty$ due to Assumption $(ii)$ listed in Theorem \ref{thm:main}. Similarly, the second term is also non-negative because the summation is bounded from below by
\begin{align}
- \frac{1}{1-\gamma}\sum_{l=0}^{k_1-2} |\ue_l| \sum_{k=k_1}^{k_2} \tbeta_k\prod_{m=l+1}^{k-1}(1-\tbeta_m) &=  - \frac{1}{1-\gamma}\sum_{l=0}^{k_1-2} |\ue_l|\left(\prod_{m=l+1}^{k_1-1}(1-\tbeta_m) - \prod_{m=l+1}^{k_2}(1-\tbeta_m)\right)\nn\\
&\geq - \frac{1}{1-\gamma}\sum_{l=0}^{k_1-2} |\ue_l|\prod_{m=l+1}^{k_1-1}(1-\tbeta_m),
\end{align}
which does not depend on $k_2$ and goes to zero as $k_1\rightarrow\infty$. Particularly, the absolute summability of $\{\ue_k\}$ and Assumption $(ii)$ yields that $\{\ue_k\}$ decays faster than $\{\beta_k\}$ and there exists $k_0$ such that $|\ue_k| \leq \beta_k$ for all $k\geq k_0$. Therefore, for $k_1\geq k_0$, we have
\begin{align}
- \frac{1}{1-\gamma}\sum_{l=0}^{k_1-2} |\ue_l|\prod_{m=l+1}^{k_1-1}(1-\tbeta_m) = &\;-\frac{1}{1-\gamma}\left(\prod_{m=k_0}^{k_1-1}(1-\tbeta_m)\right)\sum_{l=0}^{k_0-1} |\ue_l|\prod_{m=l+1}^{k_0-1}(1-\tbeta_m) \nn\\
&- \frac{1}{1-\gamma}\sum_{l=k_0}^{k_1-2} |\ue_l|\prod_{m=l+1}^{k_1-1}(1-\tbeta_m),
\end{align}
where the first-term goes to zero as $k_1\rightarrow\infty$ due to Assumption $(ii)$, and based on Lemma \ref{lem:aux}, the second term is bounded from below by
\begin{align}
- \frac{1}{1-\gamma}\sum_{l=k_0}^{k_1-2} |\ue_l|\prod_{m=l+1}^{k_1-1}(1-\tbeta_m) &\geq - \frac{1}{1-\gamma}\sum_{l=k_0}^{k_1-2} \tbeta_l\prod_{m=l+1}^{k_1-1}(1-\tbeta_m)\nn\\
&= \prod_{l=k_0+1}^{k_1-1}(1-\tbeta_l) - \prod_{l=k_1-2}^{k_1-1}(1-\tbeta_l)\nn\\
&\geq \prod_{l=k_0+1}^{k_1-1}(1-\tbeta_l),
\end{align}
which goes to zero as $k_1\rightarrow\infty$ by Assumption $(ii)$.
Finally, the third term is also non-negative because the summation is bounded from below by
\begin{align}
- \frac{1}{1-\gamma} \sum_{l=k_1-1}^{k_2-1}|\ue_l| \sum_{k=l+1}^{k_2} \tilde{\beta}_k\prod_{m=l+1}^{k-1}(1-\tilde{\beta}_m) &= - \frac{1}{1-\gamma} \sum_{l=k_1-1}^{k_2-1}|\ue_l|\left( \prod_{m=l+1}^{l}(1-\beta_l) - \prod_{m=l+1}^{k_2}(1-\beta_l)\right)\nn\\
&\geq - \frac{1}{1-\gamma} \sum_{l=k_1-1}^{k_2-1}|\ue_l|\nn\\
&\geq - \frac{1}{1-\gamma} \sum_{l=k_1-1}^{\infty}|\ue_l|,
\end{align}
which goes to zero as $k_1\rightarrow\infty$ since $\{\ue_k\}$ is absolutely summable. This completes the proof. \hfill $\square$

\section{Proof of Lemma \ref{lem:stepiii}}\label{app:stepiii}

The proof follows from \citep[Theorem 4]{ref:Leslie06}. Particularly, we can view \eqref{eq:updatenew} as a weakened fictitious play dynamic in a game with payoffs $Q_*^i(\cdot)$ for each $i$ since the action $a_k^i$ satisfies
\begin{equation}
\mathbb{E}_{a^{-i}\sim\pi_k^{-i}}\{Q^i_k(a_k^i,\pi^{-i}_k)\} = \max_{a^i\in A^i} \mathbb{E}_{a^{-i}\sim\pi_k^{-i}}\{Q_k^i(a^i,a^{-i})\} \geq \max_{a^i\in A^i} \mathbb{E}_{a^{-i}\sim\pi_k^{-i}}\{Q_*^i(a^i,a^{-i})\} - \epsilon_k
\end{equation}
for some $\epsilon_k\rightarrow 0$ as $k\rightarrow\infty$ since $\Q_k^i(a) \rightarrow Q_*^i(a)$ for all $(i,a)$ as $k\rightarrow\infty$. The asymptotic negligibility of the error term follows since 
\begin{align}
&\left|\max_{a^i\in A^i} \mathbb{E}_{a^{-i}\sim\pi_k^{-i}}\{Q_k^i(a^i,a^{-i})\} - \max_{a^i\in A^i} \mathbb{E}_{a^{-i}\sim\pi_k^{-i}}\{Q_{*}^i(a^i,a^{-i})\}\right| \nn\\
&\hspace{2in}\leq \max_{a^i\in A^i} \left|\mathbb{E}_{a^{-i}\sim\pi_k^{-i}}\{Q_k^i(a^i,a^{-i})-Q_*^i(a^i,a^{-i})\}\right|
\end{align} 
and the right-hand side goes to zero due to the convergence of $Q_k^i$ to $Q_*^i$. This completes the proof.\hfill $\square$

\section{Proof of Corollary \ref{cor:potential}}\label{app:potential}

The proof follows from the observation that based on \eqref{eq:QQ} and \eqref{eq:pp}, $\Gamma_k^{ij}$, as described in \eqref{eq:Gamma}, can be written as
\begin{align}
\Gamma_k^{ij}(s) =&\; \mathbb{E}_{a\sim\pi_k(s)}\{r^i(s,a_k^j(s),a^{-j}) - r^i(s,a)\} \sum_{l=0}^k \beta_l \left(\prod_{m=l+1}^k(1-\beta_m)\right)\nn\\
&+\gamma\sum_{s'\in S} \mathbb{E}_{a\sim\pi_k(s)}\{p(s'|s,a^i)-p(s'|s,a^i)\} \sum_{l=0}^k \vfunc_l^i(s') \beta_l \left(\prod_{m=l+1}^k(1-\beta_m)\right)\\
=&\; \mathbb{E}_{a\sim\pi_k(s)}\{r^j(s,a_k^j(s),a^{-j}) - r^j(s,a)\} \sum_{l=0}^k \beta_l \left(\prod_{m=l+1}^k(1-\beta_m)\right) \geq 0.
\end{align}
This completes the proof. \hfill $\square$

\end{spacing}

% Bibliography
\begin{spacing}{1}
\bibliographystyle{plainnat}
\bibliography{mybibfile,mybibfile_2,myref,main}
\end{spacing}

\end{document}